\documentclass[12pt]{article}
\usepackage{amsthm,amsmath,amssymb}
\usepackage{graphicx,psfrag,epsf}
\usepackage{enumerate}
\usepackage{epstopdf}
\usepackage{natbib}
\usepackage{url} % not crucial - just used below for the URL

\newcommand{\blind}{0}

\newtheorem{theorem}{Theorem}[section]
\newtheorem{lemma}{Lemma}[section]
\newcommand{\I}{{\mathbf I}}
\newcommand{\X}{{\mathbf X}}
\newcommand{\x}{{\mathbf x}}
\newcommand{\Y}{{\mathbf Y}}
\newcommand{\y}{{\mathbf y}}
\newcommand{\Z}{{\mathbf Z}}
\newcommand{\z}{{\mathbf z}}
\newcommand{\U}{{\mathbf U}}
\newcommand{\bu}{{\mathbf u}}
\newcommand{\V}{{\mathbf V}}
\newcommand{\bv}{{\mathbf v}}

\newcommand{\w}{{\mathbf w}}

\newcommand{\E}{{\rm E}}
\newcommand{\Var}{{\rm Var}}
\newcommand{\bmu}{{\boldsymbol \mu}}

\newcommand{\tr}{{\text{\rm tr}}}

\begin{document}

\def\spacingset#1{\renewcommand{\baselinestretch}%
{#1}\small\normalsize} \spacingset{1}

%%%%%%%%%%%%%%%%%%%%%%%%%%%%%%%%%%%%%%%%%%%%%%%%%%%%%%%%%%%%%%%%%%%%%%%%%%%%%%

\if0\blind
{
  \title{\bf Testing the independence of two random vectors where only one dimension is large}
  \author{Weiming Li\thanks{
    Weiming Li's research is supported by National Natural Science Foundation of China, No. 11401037, and Fundamental Research Funds for the Central Universities, No. 2014RC0905.}\hspace{.2cm}\\
    Beijing University of Posts and Telecommunications\\
    and \\
    Jiaqi Chen\thanks{
    Jiaqi Chen's research is supported by Program for Innovation Research of Science in Harbin Institute of Technology.}\hspace{.2cm} \\
    Harbin Institute of Technology\\
    and \\
    Jianfeng Yao\thanks{
    Corresponding author}\hspace{.2cm} \\
    The University of Hong Kong\\
     }
  \maketitle
} \fi

\if1\blind
{
  \bigskip
  \bigskip
  \bigskip
  \begin{center}
    {\LARGE\bf Testing the independence of two random vectors where only one dimension is large}
\end{center}
  \medskip
} \fi

\bigskip
\begin{abstract}
For testing the independence of two vectors with respective dimensions $p_1$ and $p_2$, the existing literature in high-dimensional statistics all assume that both dimensions $p_1$ and $p_2$ grow to infinity with the sample size. However, as evidenced in the RNA-sequencing data analysis discussed in the paper, it happens frequently that one of the dimension is quite small and the other quite large compared to the sample size. In this paper, we address this new asymptotic framework for the independence test. A new test procedure is introduced and its asymptotic normality is established when the vectors are normal distributed. A Mote-Carlo study demonstrates the consistency of the procedure and exhibits its superiority over some existing high-dimensional procedures. Applied to the RNA-sequencing data mentioned above, we obtain very convincing results on pairwise independence/dependence of gene isoform expressions as attested by prior knowledge established in that field. Lastly, Monte-Carlo experiments show that the procedure is robust against the normality assumption on the population vectors.

\end{abstract}

\noindent%
{\it Keywords:}  Covariance matrix; Gene network; High-dimensional testing; Independence test.
\vfill

\newpage
\spacingset{1.45} % DON'T change the spacing!
\section{Introduction}
\label{sec:intro}

Modern scientific researches increasingly encounter high dimensional data and then evoke corresponding statistical analyses.
In genomics, next-generation sequencing techniques such as
RNA-Sequencing \citep{Feng13} are designed to quantify gene
expression, where typically a group of gene isoforms are analyzed and
their expression data at exon levels are recorded into
multidimensional vectors. The dimensions of these vectors vary in a
wide range where the smallest dimension can be one or two and the
largest one can be comparable to the sample size (see Table 3).  A fundamental issue in such analyses is determining whether there is any interaction between two given gene isoforms. More formally, this problem involves testing the independence of two possibly correlated vectors in a situation where one dimension is small but the other is large compared to the sample size.

Generally, let
$\X=(X_1,\ldots,X_{p_1})$,
$\Y=(Y_1,\ldots,Y_{p_2})$
and $\Z=(\X,\Y)$  be
the joint vector of dimension $p:=p_1+p_2$.
The covariance matrix  of $\Z$ is  partitioned as
\begin{equation*}\label{partition}
  \Sigma=
  \left(\begin{matrix}
      \Sigma_{xx}&\Sigma_{xy}\\
      \Sigma_{yx}&\Sigma_{yy}
    \end{matrix}\right)
\end{equation*}
so that $\Sigma_{xx}=Var(\X)$, $\Sigma_{yy}=Var(\Y)$ and
$\Sigma_{xy}=Cov(\X,\Y)$.
%The matrix $\Sigma$ is hereafter referred as the population covariance matrix (PCM).
Let $\z_1,\ldots,\z_N$ be a sample of size $N$
drawn from the population $\Z$.
The sample covariance matrix is
\begin{eqnarray*}\label{sn}
  S_n=\frac{1}{n}\sum_{k=1}^N(\z_k-\bar\z)(\z_k-\bar\z)'
\end{eqnarray*}
where $\bar\z=\frac{1}{N}\sum_{k=1}^N\z_k$ and $n=N-1$ represents the degree of freedom.
Accordingly, $S_n$ can be partitioned as
\[
S_n=\left(\begin{matrix}
    S_{xx}&S_{xy}\\
    S_{yx}&S_{yy}
  \end{matrix}\right).
\]
Assume that the joint vector $\Z$
has a  $p$-dimensional normal distribution with mean $\bmu$
and covariance matrix $\Sigma$,
the independence hypotheses of $\X$ and $\Y$ can be represented as
\begin{equation}\label{hyp}
H_0: \Sigma_{xy}=0\quad v.s.\quad H_1: \Sigma_{xy}\neq 0.
\end{equation}
To test these hypotheses,  the following three statistics are
commonly used \citep{Anderson03}, which are the likelihood ratio test (LRT) and two trace criteria:
\begin{eqnarray}
 &&\Lambda=\frac{\sup_{H_0} L(\bmu,\Sigma)}{\sup L(\bmu,\Sigma)}
    =\frac{|S_n|^{N/2}}{|S_{xx}|^{N/2}|S_{yy}|^{N/2}}=|\I_{p_1}-S_{xy}S_{yy}^{-1}S_{yx}S_{xx}^{-1}|^{\frac{N}{2}},\nonumber\\
 &&  C_1=\tr( S_{xy}S_{yy}^{-1}S_{yx}S_{xx}^{-1})\quad\text{and}\quad C_2=\tr (S_{xy}S_{yx}) - \frac1n \tr(S_{xx})  \tr(S_{yy}).\label{t12}
\end{eqnarray}
The LRT statistic is the well-known  Wilks's $\Lambda$
\citep{W35}. Both  statistics $C_1$  and $C_2$ are  based on the idea that under
the independence hypothesis,  $\Sigma_{xy}=\Sigma'_{yx}=0$ so that
$S_{xy}$ as well as $S_{yx}$ should be
small. A noticeable difference here is that the statistics
$\Lambda$ and $C_1$ rely on the inverse matrices
$S_{xx}^{-1}$ and $S_{yy}^{-1}$ so that essentially the conditions
$p_i<n$ are required.
Conversely, the criterion $C_2$ can be applied when the dimensions
$p_i$, $i=1,2$, are  larger than the sample size $N$.
%Notice that all three statistics are functions of eigenvalues of
%the sample covariance matrices $S_{xx}$, $S_{yy}$ or their product
%(matrix $M$).

The test procedures for the classical situation where
the dimensions $p_i$'s are reasonably small compared with the the sample size are well studied \citep{Anderson03}.
It is however well understood today that
these asymptotical approximations are no more valid  when the dimensions
are comparable to the sample size, see e.g.  \citet{LW02}, \citet{BJYZ09},
\citet{ChenQin10} and \citet{WY13}.
New limiting distributions have to be found in the
large-dimensional context.

Specifically for the independence test, the existing literature in the
large-dimensional context includes
\begin{enumerate}
\item
  the
  large-dimensional limit of $\Lambda$ proposed
  in \citet{JBZ13} under the asymptotic scheme
  $\min(p_1,p_2,n)\to \infty$, $p_1+p_2<n$ and  $p_i/n\to c_i>0$;
\item
  the large-dimensional limit of $C_1$  proposed
  in \citet{JBZ13} under the  asymptotic scheme
  $\min(p_1,p_2,n)\to \infty$, $\max(p_1,p_2)<n$ and  $p_i/n\to
  c_i>0$; and
\item
  the  large-dimensional limit of $C_2$  proposed
  in \citet{SR12}  under the  asymptotic scheme
  $\min(p_1,p_2,n)\to \infty$,   $p_i/p\to
  d_i>0$ and   $n=O(p^\delta)$ for some constant $\delta>0$
  as $n\rightarrow\infty$.
\end{enumerate}
 These existing  asymptotic schemes
are quite similar in that they all require
that both dimensions $p_1$ and $p_2$ grow  to infinity
with the sample size $N$.

Motivated by RNA-sequencing analysis, our objective in this paper is to test the hypotheses in \eqref{hyp} with the criterion $C_2$ assuming $p_1$ fixed and $(p_2,n)\rightarrow\infty$. As far as we know, this scheme has not been addressed in the literature. It will be proved that the asymptotic distribution of the statistic exists under this asymptotic scenario and is consistent with the one in \cite{SR12}. Note that our proof is different from theirs and this new asymptotic scenario is not covered by their results.

The rest of this paper is organised as follows. In the next section, we present the new test procedure and examine its size and power through simulation experiments. Section 3 presents an analysis of a genomic data set and Section 4 presents some conclusions and remarks. The main theorem is proved in the last section.

\section{Test for the independence in high dimensions}
\label{sec:test}
\subsection{Test statistic and its asymptotic distribution}

The null hypothesis in \eqref{hyp} is equivalent to $\tr(\Sigma_{xy}\Sigma_{yx})=0.$ Thus we may construct an unbiased estimator of this trace and reject the null hypothesis when this statistic is too large. Let
\begin{equation*}\label{gamma}
\gamma_{2}=\tr(\Sigma^2),\quad \gamma_{xx}=\tr(\Sigma_{xx}^2), \quad \gamma_{yy}=\tr(\Sigma_{yy}^2),\quad \gamma_{xy}=\tr(\Sigma_{xy}\Sigma_{yx}).
\end{equation*}
We have by definition $2\gamma_{xy}=\gamma_{2}-\gamma_{xx}-\gamma_{yy}$. From \cite{Srivastava05}, an unbiased estimator of $\gamma_2$ is given as $k_n[\tr(S^2_n)-\tr^2(S_n)/n]$ with $k_n=n^2/(n-1)(n+2)$. Therefore an unbiased estimator of $\gamma_{xy}$ is constructed as
\begin{eqnarray*}
\hat \gamma_{xy}&=&\frac{k_n}{2}\left\{\tr(S^2_n)-\tr(S_{xx}^2)-\tr(S_{yy}^2)
-\frac{1}{n}\left[\tr^2(S_n)-\tr^2(S_{xx})-\tr^2(S_{yy})\right]\right\},\nonumber\\
&=&k_n\left[\tr(S_{xy}S_{yx})-\frac{1}{n}\tr(S_{xx})\tr(S_{yy})\right].\label{test}
\end{eqnarray*}
We thus get the trace criterion $C_2$ given in \eqref{t12}. Notice that the estimator $\hat \gamma_{xy}$ is a function of eigenvalues of the sample covariance matrices $S_{xx}$, $S_{yy}$, and $S_n$.

\begin{theorem}\label{th1}
Suppose that the dimensions $p=p_1+p_2$ and $n$ both tend to infinity, and
$$0<\lim_{p\rightarrow\infty}\frac{1}{p} \tr(\Sigma^k)< \infty,\ k=1,2,4.$$
Then under the null hypothesis in \eqref{hyp},
\begin{equation}\label{clt}
T_n:=\frac{n}{\sqrt{2k_n}}\frac{\hat \gamma_{xy}}{\sqrt{\hat \gamma_{xx}\hat \gamma_{yy}}}\xrightarrow{d} N(0,1),
\end{equation}
where $\hat \gamma_{xx}=k_n[\tr(S^2_{xx})-\tr^2(S_{xx})/n]$ and $\hat \gamma_{yy}=k_n[\tr(S^2_{yy})-\tr^2(S_{yy})/n]$ with $k_n=n^2/(n-1)(n+2)$.
\end{theorem}
This theorem is built on a general dimensional scenario as only the assumption $p_1+p_2\rightarrow\infty$ is required. This scenario integrates two cases: 1) $p_1$ is fixed and only $p_2$ approaches infinity; 2) $p_1$ and $p_2$ both tend to infinity. Under the second case, the conclusion in \eqref{clt} is essentially the same as the main theorem in \cite{SR12}. This means that for practical applications,  the proposed test is robust against different asymptotic scenarios of dimensions. Such robustness is especially welcomed since in a precise application (such as the gene isoform data analyzed in the paper) the explicit values of the dimensions $p_1$ and $p_2$ are known and it is somehow difficult to decide what is the most convenient asymptotic scenario to use.

\subsection{Monte-Carlo study}

We numerically evaluate the finite-sample performance of the test $T_n$ and report the empirical size and power under different dimension settings. For the purpose of comparison, we also consider two tests discussed in \cite{JBZ13}: one is the corrected LRT, referred as $T_1$, and the other is based on $\tr(S_{xy}S_{yy}^{-1}S_{yx}S_{xx}^{-1})$, referred as $T_2$. Since the test $T_{1}$ is limited to $p_1+p_2<n$ and $T_2$ is limited to $\max\{p_1, p_2\}<n$, we only consider the former case when comparing the three tests. The nominal significance level is fixed at $\alpha=0.05$, and the number of independent replications is 100, 000.

We first report the empirical sizes of the three tests. Samples are drawn from standard normal population, and thus $\Sigma$ is an identity matrix. The dimensions are $p_1=2,6,10$, $p_2=10,30,100,200,500$, and $n=50$. The results are collected in Table \ref{table1}, where the first six columns compare the sizes of the three tests when $p_1+p_2<n$ and the last three columns illustrate the size of the proposed $T_n$ when $p_2>n$. The results show that all the empirical sizes are close to the nominal significance level.

\begin{table}[h] \setlength\tabcolsep{6pt}
\begin{center}
\caption{Empirical sizes in percents for the three tests with the significant level $\alpha=0.05$.}
\begin{tabular}{cccccccccccccccccccc}
\hline
        $p_1$    & \multicolumn{3}{c}{$p_2=10$} &&\multicolumn{3}{c}{$p_2=30$} &&\multicolumn{3}{c}{$T_n\& p_2$}\\
                 & $T_{n}$ &$T_{1}$&$T_{2}$&&$T_{n}$ &$T_{1}$&$T_{2}$&&100&200&500\\
            \cline{2-4} \cline{6-8}\cline{10-12}
       $2$    &6.32&6.56&5.86&&5.72&6.17&4.48&&5.52 &5.34 &5.30\\
       $6$    &5.89&6.11&5.37&&5.66&5.88&4.69&&5.46 &5.21 &5.29\\
       $10$  &5.74&6.03&5.27&&5.46&5.90&4.70&&5.36 &5.09 &5.16\\
\hline
\end{tabular}
\label{table1}
\end{center}
\end{table}

To examine the powers of the three tests, we employ a model studied in \cite{JBZ13}, where the populations $\X$ and $\Y$ are defined as
\begin{equation*}
\X=\U_1+\gamma \U_2^{p_1},\quad \Y=\U_2+\gamma \U_2,\quad \U_i\sim N(0, \I_{p_i}),\quad i=1,2,
\end{equation*}
respectively, where $\U_1$ and $\U_2$ are independent, $\U_2^{p_1}$ is a subset of $\U_2$ consisting of its first $p_1$ variables, and the factor $\gamma$ represents the degree of mixture.
Therefore, the covariance matrices are respectively
$$
\Sigma_{xx}=(1+\gamma^2)I_{p_1},\quad \Sigma_{yy}=(1+\gamma)^2I_{p_2}, \quad \Sigma_{xy}=\gamma(1+\gamma)(I_{p_1}, O_{p_1,p_2-p_1}),
$$
where $O_{m,n}$ represents an $m\times n$ zero matrix.

\begin{figure}[h]
\begin{minipage}[t]{0.5\linewidth}
\includegraphics[width=2.5in]{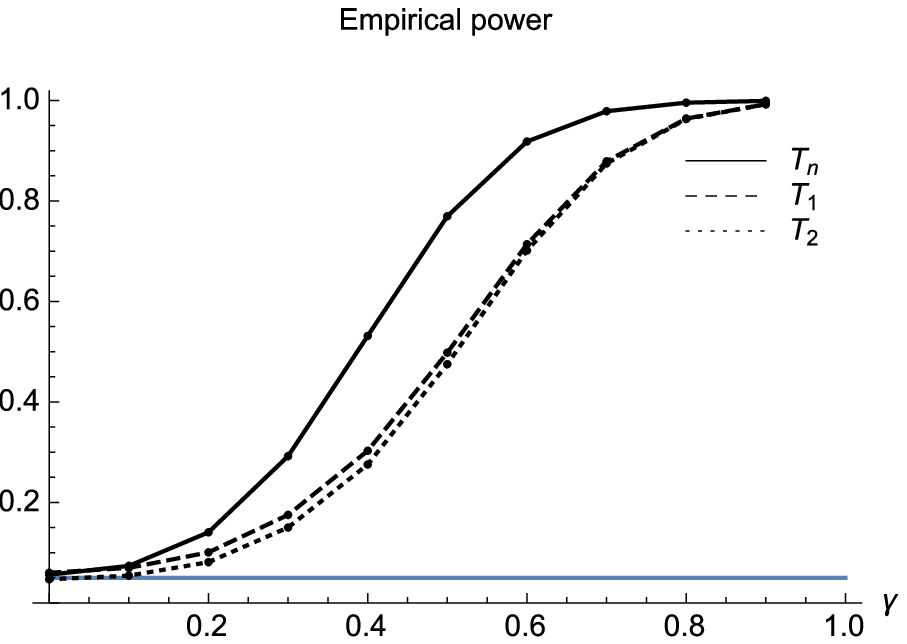}
\end{minipage}%
\begin{minipage}[t]{0.5\linewidth}
\includegraphics[width=2.5in]{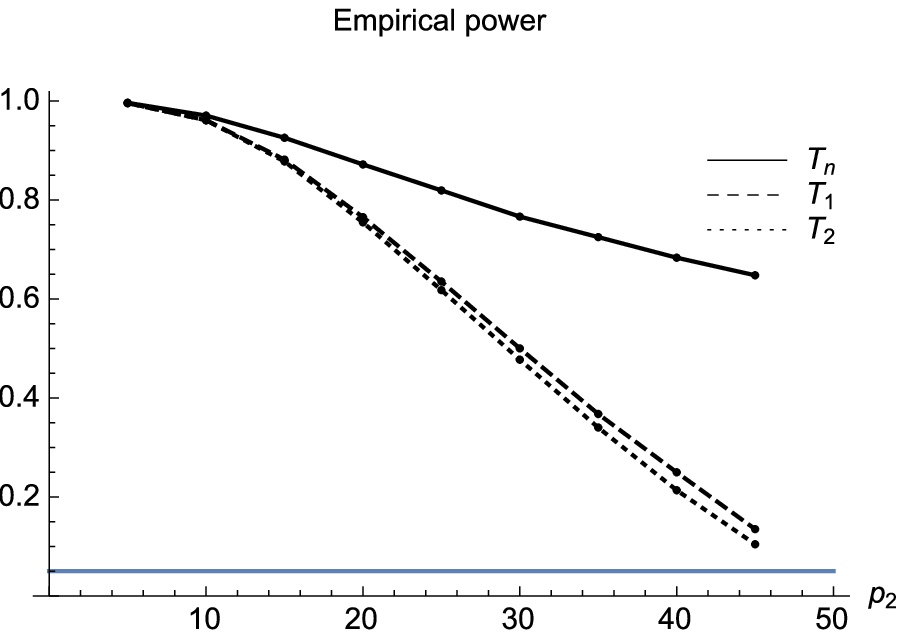}
\end{minipage}
\caption{Empirical powers of the three tests. The parameter settings are $(p_1,p_2,n)=(4,30,50)$, $0\leq \gamma\leq 0.9$ in the left panel, and $(p_1,n, \gamma)=(4,50, 0.5)$, $5\leq p_2\leq 45$ in the right panel.}
\label{fig1}
\end{figure}

Figure \ref{fig1} illustrates the powers of the three tests for this model. In the left panel, the parameters are $(p_1,p_2,n)=(4,30,50)$ and the factor $\gamma$ increases from 0 to 0.9; while on the right, $(p_1,n,\gamma)=(4,50,0.5)$ and $p_2$ increases from 5 to 45. The curves in the figure show that the powers of the tests $T_1$ and $T_2$ are similar, and are dominated by the proposed test $T_n$ in all the settings. Particularly, the curves in the right panel show that all the powers of the tests decrease as $p_2$ increases, which reflects the fact that in this process the increasing number of zero entries of $\Sigma_{xy}$ makes it closer to the zero matrix of the null hypothesis. However, the power of $T_n$ declines much slower than $T_1$ and $T_{2}$, which demonstrates a greater robustness of $T_n$ against the inflating $p_2$.

Next we examine the robustness of the three test procedures when the assumed normal distribution of the vectors is contaminated by gamma-distributed errors. The studied model is the same as the previous except that the vector $\U_i$'s are replaced by
$$\U_i+\theta\V_i,\quad
\V_i=(v_{i1},\dots,v_{ip_i})',\quad i=1,2,
$$
where $\{v_{ij}\}$, independent of $\{\U_i\}$, are $i.i.d.$ standardized random variables derived from $Gamma(a,b)$ distributed variables and the parameter $\theta$ represents the level of contamination. The new parameters are set to be $a=b=3$ (positive skew, heavy-tailed) and $ \theta=1/2, 2$ in this experiment. Thus the covariance matrices become
$$
\Sigma_{xx}=(1+\gamma^2)(1+\theta^2)I_{p_1},\quad \Sigma_{yy}=(1+\gamma)^2(1+\theta^2)I_{p_2}, \quad \Sigma_{xy}=\gamma(1+\gamma)(1+\theta^2)(I_{p_1}, O_{p_1,p_2-p_1}).
$$

Results about the empirical sizes and powers of the tests are collected in Table \ref{table2} and Figure \ref{fig2}, respectively. It shows that all the sizes are close to the nominal one and the power curves are quite similar to those in Figure \ref{fig1}, which demonstrate that the additional gamma-distributed errors have little impact on the three tests. It is however worth noticing that the theoretic proof of Theorem \ref{th1} in this paper as well as the proofs for asymptotic normality of the test criteria $T_1$ and $T_2$ established in \cite{JBZ13} all heavily rely on the assumed normality of the vectors, and to our best knowledge, it seems unclear how these proofs can be extended to cover non-normal data as the ones tested  in the Monte-Carlo experiments reported here.

\begin{table}[h] \setlength\tabcolsep{6pt}
\begin{center}
\caption{Empirical sizes in percents for the three tests with the significant level $\alpha=0.05$.}
\begin{tabular}{cccccccccccccccccccc}
\hline
          &$p_1$    & \multicolumn{3}{c}{$p_2=10$} &&\multicolumn{3}{c}{$p_2=30$} &&\multicolumn{3}{c}{$T_n\& p_2$}\\
          &       & $T_{n}$ &$T_{1}$&$T_{2}$&&$T_{n}$ &$T_{1}$&$T_{2}$&&100&200&500\\
            \cline{3-5} \cline{7-9}\cline{11-13}
                                    &$2$    &6.38&6.56&5.90&&5.86&6.19&4.45&&5.55 &5.32 &5.22\\
    $\theta=\frac{1}{2}$ &$6$    &5.91&6.17&5.44&&5.47&5.84&4.60&&5.34 &5.24 &5.16\\
                                    &$10$  &5.71&5.94&5.18&&5.55&5.80&4.80&&5.33 &5.12 &5.22\\
\noalign{\smallskip}\hline\noalign{\smallskip}
                                    &$2$    &6.38&6.52&5.80&&6.00&6.38&4.62&&5.59 &5.59 &5.33\\
    $\theta=2$               &$6$    &6.02&6.15&5.49&&5.65&5.79&4.67&&5.42 &5.42 &5.22\\
                                    &$10$  &5.85&6.03&5.27&&5.68&5.82&4.84&&5.35 &5.33 &5.06\\
\hline
\end{tabular}
\label{table2}
\end{center}
\end{table}

\begin{figure}[h]
\begin{minipage}[t]{0.5\linewidth}
\includegraphics[width=2.5in]{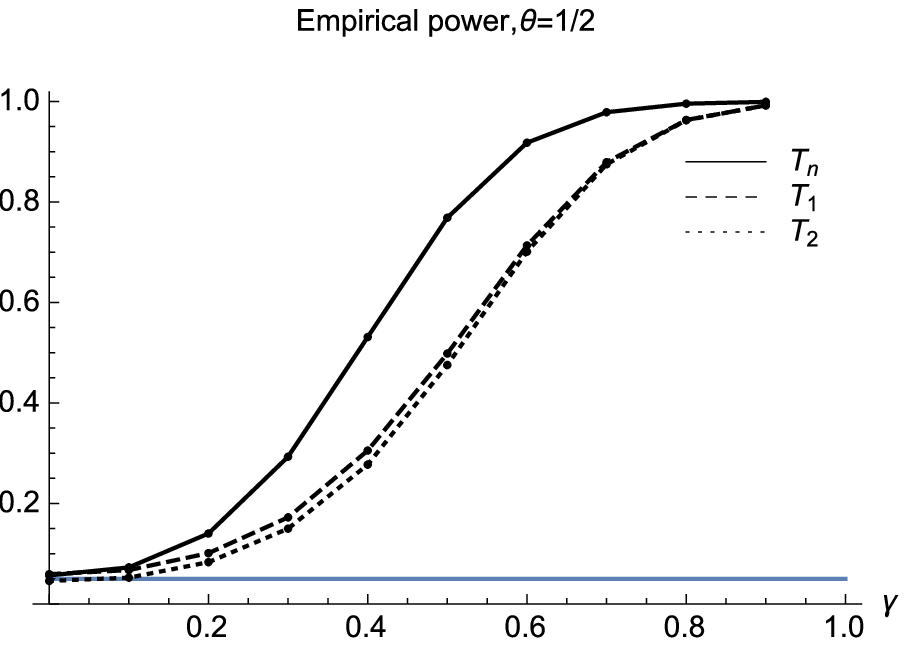}
\end{minipage}%
\begin{minipage}[t]{0.5\linewidth}
\includegraphics[width=2.5in]{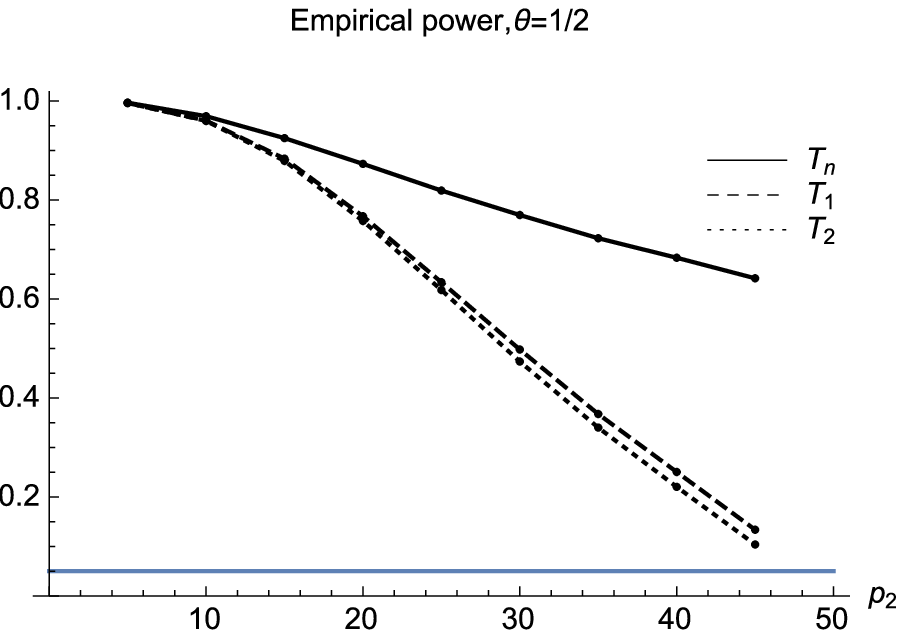}
\end{minipage}
\begin{minipage}[t]{0.5\linewidth}
\includegraphics[width=2.5in]{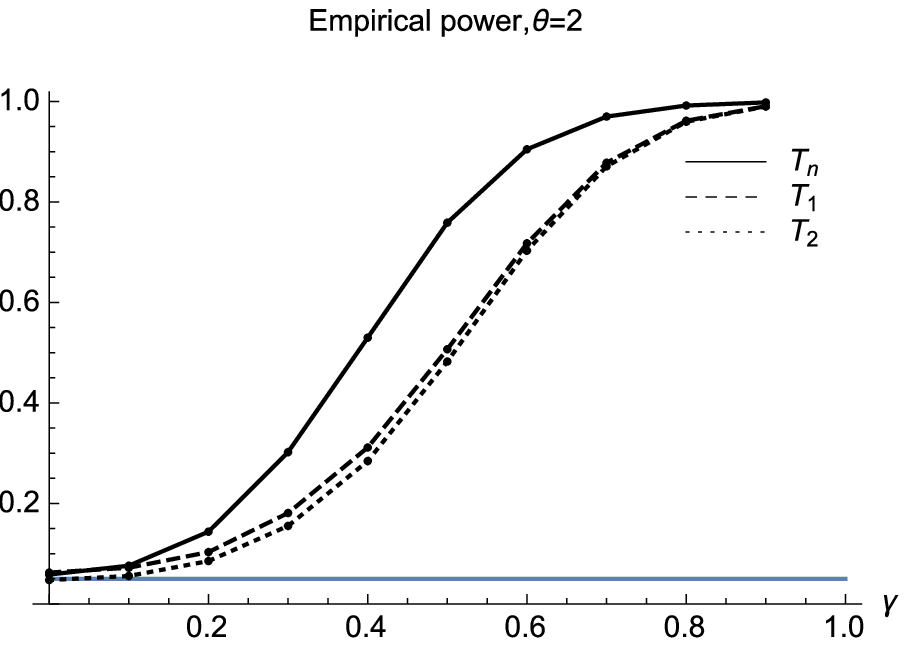}
\end{minipage}%
\begin{minipage}[t]{0.5\linewidth}
\includegraphics[width=2.5in]{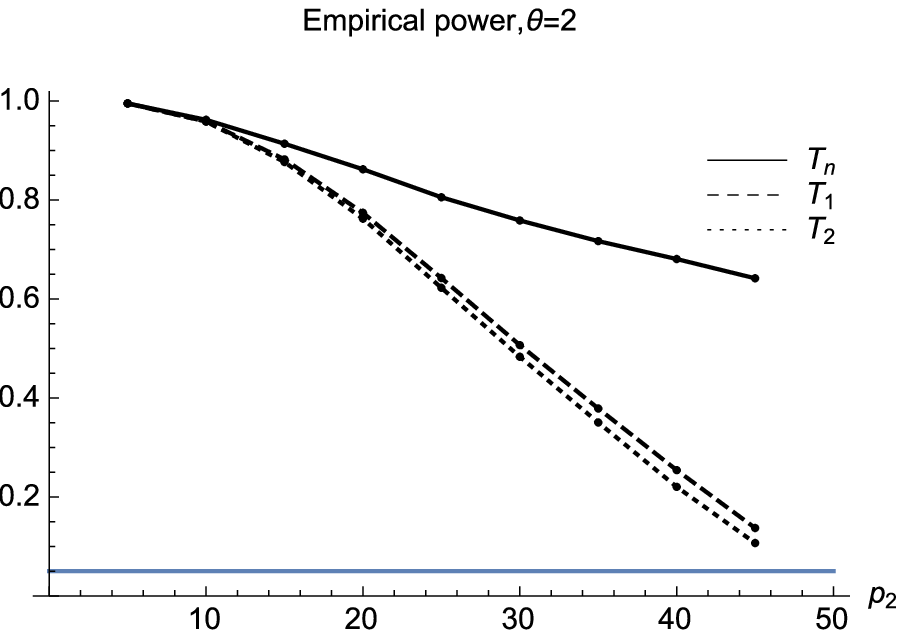}
\end{minipage}
\caption{Empirical powers of the three tests for the non-normal distribution with $\theta=1/2, 2$. The parameter settings are $(p_1,p_2,n)=(4,30,50)$, $0\leq \gamma\leq 0.9$ in the left panel, and $(p_1,n, \gamma)=(4,50, 0.5)$, $5\leq p_2\leq 45$ in the right panel.}
\label{fig2}
\end{figure}

\section{Real data analysis}
\label{sec:realdata}
Genomes play a central role in the control of cellular processes \citep{BO04}. The dynamic interplay between various genes can be mapped as gene co-expression networks, which is an important and widely used method to understand the cause and prognosis of various diseases. To recover pairwise dependencies in a gene co-expression network, each co-expression edge has to be inferred by accepting or rejecting the independence hypothesis from the sample covariance matrix of respective isoform expressions.

We analyze a data set of liver cancer, which is downloaded from TCGA data portal: https://tcga-data.nci.nih.gov/tcga/dataAccessMatrix.htm, and filtered by data types RNASeqV2 and Level 3. The data set consists of 38 genes with their dimensions ranging from 1 to 31 (see Table \ref{genes}) and their sample size is $N=50$. Obviously the dimensions are not on the same order of magnitude as their sample size. For these genes, the relationship of dependency are totally known based on established knowledge from historical experiments: 29 pairs of them are dependent and the remaining 674 pairs are independent.
\begin{table}[h]
\footnotesize
\caption{Lung cancer data: 38 genes with different dimensions}
\centering
\begin{tabular}{|c|cccccc|}
\hline
Name         & NM000222&NM000321&NM000636& NM000791&NM001126116& NM001140\\
Dimension    & 20      & 27     &4       & 6       & 7         & 13      \\
\hline
Name         & NM001145102&NM001204191&NM001237&NM001429& NM001759& NM001760\\
Dimension    & 9          & 7         & 8      & 31     & 5       & 4       \\
\hline
Name         &NM001786& NM001880&  NM001950& NM002198& NM002228& NM002421\\
Dimension    & 4      & 13      &  10      & 9       & 1       & 10      \\
\hline
Name         & NM002467&NM002505&NM002539&NM002985&NM003109&NM003153\\
Dimension    & 3       & 10     & 11     & 3      & 6      & 22     \\
\hline
Name         &  NM003221&NM003998&NM004379&NM004417& NM005194& NM005238\\
Dimension    & 7        & 23     &  8     & 2      & 1       & 8       \\
\hline
Name          & NM005239& NM005252&NM005438&NM007122& NM022457&NM033285\\
Dimension     & 10      & 2       & 4      & 11     & 20      & 4      \\
\hline
Name          &NM053056& NM198253&&&&\\
Dimension     & 5      &15       &&&&\\
\hline
\end{tabular}
\label{genes}
\end{table}

We test the pairwise gene dependencies using $T_n$ and compare the results with those from two other methods: one is from \cite{HCJX13}, which is a variant of traditional canonical correlation analysis (CCA); the other is the large-dimensional trace criterion $T_2$, which is recently applied in \cite{YLW14} and is demonstrated better than CCA. The corrected LRT $T_1$ is excluded from this comparison since its dimensional requirement is not met for the data set. The significance level is set to be $\alpha=0.05$. To evaluate the accuracy of the test results, we employ the so called F-score \citep{P07}
which actually measures the trade-off between precision $P$ and recall $R$:
\begin{eqnarray}
F=2\times \frac{P\times R}{P+R},
\end{eqnarray}
where
\begin{eqnarray*}
P=\frac{true \ positives}{true \ positives + false \ positives},\quad R=\frac{true \ positives}{true \ positives + false \ negatives}.
\end{eqnarray*}
With the prior information of dependency, the {\em true positives} stands for the number of correctly identified correlated pairs of genes, the {\em false positive} is the number of misidentified correlated pairs of genes, and the {\em false negatives} is the number of misidentified uncorrelated pairs of genes.

The F-scores reported in Table \ref{F1} show that $T_n$ outperforms $T_2$ significantly. CCA fails to detect the relationship between gene NM002228 and other genes due to the dimension of this gene is 1. The same phenomenon happens to gene NM005195. Therefore, we cannot get F-score for CCA.

\begin{table}[h]
\caption{F-scores for the data set including 38 genes.}
\centering
\begin{tabular}{cccc}
\hline
Method & $T_n$ & $T_2$ & CCA\\
F-score & 0.64 & 0.40 & NA\\
\hline
\end{tabular}
\label{F1}
\end{table}

Next, we remove the 1-dimensional genes from the data set in order to incorporate CCA for comparison. The remaining 36 genes include 25 dependent pairs and 605 independent pairs.  The F-scores collected in Table \ref{F2} demonstrate that $T_n$ again outperforms the others. Notice that such results on pairwise dependence of gene isoform expressions are further used to construct gene co-expression networks, see \cite{YLW14}.

\begin{table}[h]
\caption{F-scores for the data set including 36 genes.}
\centering
\begin{tabular}{cccc}
\hline
Method & $T_n$ & $T_2$ & CCA \\
F-score & 0.6465 & 0.4238 & 0.4187 \\
\hline
\end{tabular}
\label{F2}
\end{table}

\section{Concluding remarks}
\label{sec:conc}
This paper investigates the independence test of two vectors in a high-dimensional situation where one of the dimensions $p_1$ is quite small while the other dimension $p_2$ is large compared to the sample size. The asymptotic scheme is novel and practically useful. A new procedure is introduced and the test statistic under the null is proved to be asymptotically normal distributed assuming that $p_1+p_2\rightarrow\infty$ and the vectors are normal distributed. The power of the proposed test is studied through Monte-Carlo simulations and a real data analysis, which demonstrates the superiority of the new test over the existing ones. Another interesting feature found in the Monte-Carlo study is that the proposed procedure is robust against deviations from the normality assumption on the vectors although a theoretic proof of this fact is still missing.

\section{Proofs}

\subsection{Lemma}
\begin{lemma}\label{lemma1}
Let $\bu$, $\bv$, and $\w$ be independent vectors of $n$-dimensional standard normal distribution $N(0, I_n)$, and  define
\begin{equation}\label{psi}
\psi(\x, \y)=\frac{1}{n}(\x'\y)^2-\frac{1}{n^2}(\x'\x)(\y'\y),
\end{equation}
then
\begin{eqnarray*}
&& \E[\psi(\bu, \bv)|\bu]=0,\quad\E[\psi(\bu, \bv)\psi(\w,\bv)|\bu,\w]=\frac{2}{n}\psi(\bu,\w),\\
&& \E[\psi(\bv, \bv)]=(n-1)(n+2)/n,\quad\E[\psi^2(\bu, \bv)]=2(n-1)(n+2)/n^2,\\
&& \E[\psi^2(\bv, \bv)]=O(n^2),\quad \Var[\psi^2(\bv, \bv)]=O(n), \quad\E[\psi^4(\bu,\bv)]=O(1),
\end{eqnarray*}
as $n\rightarrow\infty.$
\end{lemma}
\begin{proof}
The distribution of $\bv'\bv$ is $\chi^2(n)$ and the conditional distribution of $\bu'\bv|\bu$ is $N(0,\bu'\bu)$,
thus ${\rm E}[\psi(\bu, \bv)|\bu]=0$. Write
\begin{eqnarray*}
\psi(\bu, \bv)\psi(\w, \bv)&=&\frac{1}{n^2}(\bu'\bv)^2(\w'\bv)^2-\frac{1}{n^3}(\bu'\bv)^2(\w'\w)(\bv'\bv)\\
&&-\frac{1}{n^3}(\w'\bv)^2(\bu'\bu)(\bv'\bv)+\frac{1}{n^4}(\bu'\bu)(\w'\w)(\bv'\bv)^2\\
&:=&\frac{1}{n^2}S_1-\frac{1}{n^3}S_2-\frac{1}{n^3}S_3+\frac{1}{n^4}S_4.
\end{eqnarray*}
Then $\E(S_4|\bu,\w)=n(n+2)(\bu'\bu)(\w'\w)$, and
\begin{eqnarray*}
\E(S_1|\bu,\w)
&=&\sum_{i,j,k,l}u_iu_jw_kw_l\E(v_iv_jv_kv_l)\\
&=&\sum_{i=j, k=l}u_iu_jw_kw_l+\sum_{i=k, j=l}u_iu_jw_kw_l+\sum_{i=l, j=k}u_iu_jw_kw_l\\
&=&(\bu'\bu)(\w'\w)+2(\bu'\w)^2,\\
\E(S_2|\bu,\w)&=&(\w'\w)\sum_{i,k}u_i^2\cdot\E(v_i^2v_k^2)=(n+2)(\bu'\bu)(\w'\w),
\end{eqnarray*}
and thus $\E(S_3|\bu,\w)=\E(S_2|\bu,\w)$, where $(x_i)$ denote the elements of $\x$.
Collecting these results, we get $\E[\psi(\bu, \bv)\psi(\w,\bv)|\bu,\w]=(2/n)\psi(\bu,\w)$.

Notice that $\psi(\bv,\bv)=(n-1)(\bv'\bv)^2/n^2$, and $\E(\bv'\bv)^k=n(n+2)\cdots(n+2k-2)$, $k\in\mathbb N^+$.
We have then,
\begin{eqnarray*}
\E[\psi(\bv,\bv)]&=&(n-1)(n+2)/n,\\
\E[\psi^2(\bu, \bv)]&=&(2/n)\E[\psi(\bu,\bu)]=2(n-1)(n+2)/n^2,\\
\E[\psi^2(\bv, \bv)]&=&\E(\bv'\bv)^4(n-1)^2/n^4=O(n^2),\\
\Var[\psi^2(\bv, \bv)]&=&\left[\E(\bv'\bv)^4-\E^2(\bv'\bv)^2\right](n-1)^2/n^4=O(n).
\end{eqnarray*}
Finally, from Minkowski inequality,
\begin{eqnarray*}
\E[\psi^4(\bu,\bv)]&=&\frac{1}{n^4}\E[(\bu'\bv)^2-(\bu'\bu)(\bv'\bv)/n]^4\\
&\leq&\frac{1}{n^4}\left\{\left[\E(\bu'\bv)^8\right]^{\frac{1}{4}}+\left[\E(\bu'\bu)^4\E(\bv'\bv)^4\right]^{\frac{1}{4}}/n\right\}^4\\
&=&\frac{1}{n^4}\left\{\left[\E(\bv'\bv)^4\right]^{\frac{1}{4}}+\left[\E(\bv'\bv)^4\right]^{\frac{1}{2}}/n\right\}^4,
\end{eqnarray*}
which is $O(1)$ as $n\rightarrow\infty.$
\end{proof}

\subsection{Proof of Theorem \ref{th1}}

The sample covariance $S_n$ has the Wishart distribution $W_n(\Sigma)$ with $n$ degrees of freedom. It can be expressed as $\sum_{i=1}^{n}\tilde \z_k\tilde \z_k'/n$ where ($\tilde \z_i$) are i.i.d. $N(0,\Sigma)$.
Write $\tilde \z_i=(\tilde \x_{i}', \tilde \y_{i}')'=(\tilde x_{i1},\ldots,\tilde x_{ip_1}, \tilde y_{i1},\ldots,\tilde y_{ip_2})'$, $i = 1,\ldots{n}$, and denote $\X=(\tilde\x_1,\ldots,\tilde\x_{n})$ and $\Y=(\tilde\y_1,\ldots,\tilde\y_{n})$.
Note that the matrices $\X$ and $\Y$ contain normal vectors which are independent under $H_0$. The matrices $\X'\X$ and $\Y'\Y$ can be standardized as
\begin{eqnarray*}
\X'\X=\sum_{i=1}^{p_1}\alpha_i\bu_i\bu_i',\quad\Y'\Y=\sum_{j=1}^{p_2}\beta_j\bv_j\bv_j',
\end{eqnarray*}
where $(\alpha_i)$ and $(\beta_j)$ are the eigenvalues of $\Sigma_{xx}$ and $\Sigma_{yy}$, respectively, and $(\bu_i), (\bv_j)$ are i.i.d. $N(0, I_n)$.
Therefore, we have
\begin{eqnarray*}
\frac{n}{k_n}\hat\gamma_{xy}&=&n\tr(S_{xy}S_{yx})-\tr(S_{xx})\tr(S_{yy})\\
&=&\frac{1}{n}\tr(\X'\X\Y'\Y)-\frac{1}{n^2}\tr(\X'\X)\tr(\Y'\Y)\\
&=&\sum_{i=1}^{p_1}\sum_{j=1}^{p_2}\alpha_i\beta_j\left[\frac{1}{n}(\bu_i'\bv_j)^2
-\frac{1}{n^2}(\bu_i'\bu_i)(\bv_j'\bv_j)\right]\\
&:=&\sum_{i=1}^{p_1}\sum_{j=1}^{p_2}a_{ij}\psi(\bu_i,\bv_j),\label{asy1}
\end{eqnarray*}
where $a_{ij}=\alpha_i\beta_j$ and $\psi$ is defined in \eqref{psi} with the dimension n, $i=1,\ldots,p_1,\ j=1,\ldots,p_2.$

We use the martingale CLT to establish the limiting distribution of $T_n$. Without loss of generality suppose that $p_1\leq p_2$, and define
$\phi_j^{(n)}=(1/\sqrt{p_1p_2})\sum_{i=1}^{p_1}a_{ij}\psi(\bu_i,\bv_j)$, $j=1,\ldots,p_2$. Let $\mathcal F_j^{(n)}$ be the $\sigma$-algebra generated by the random variables $\{\bu_1,\ldots,\bu_{p_1},\bv_1,\ldots,\bv_{j}\}$, then $\{\emptyset, \Omega\}=\mathcal F_0\subset\mathcal F_1^{(n)}\subset\cdots\subset \mathcal F_{p_2}^{(n)}\subset \mathcal F$ with $(\Omega, \mathcal F, P)$ the probability space. From Lemma \ref{lemma1} and the law of iterated expectations,
\begin{eqnarray*}
\E\left[\phi_j^{(n)}\bigg|\mathcal F_{j-1}^{(n)}\right]
&=&\frac{1}{\sqrt{p_1p_2}}\sum_{i=1}^{p_1}a_{ij}\E(\psi(\bu_i,\bv_j)|\bu_i)=0,\\
\E\left[\phi_j^{(n)}\right]^2
&=&\frac{1}{p_1p_2}\sum_{i=1}^{p_1}\sum_{k=1}^{p_1}a_{ij}a_{kj}\E[\psi(\bu_i,\bv_j)\psi(\bu_k,\bv_j)]\\
&=&\frac{2(n-1)(n+2)}{n^2p_1p_2}\sum_{i=1}^{p_1}a_{ij}^2,
\end{eqnarray*}
which is $O(1/p_2)$ as $(p, n)\rightarrow\infty$. Thus $\{\psi_j^{(n)}, \mathcal F_j^{(n)}\}$ forms a sequence of integrable martingale differences.
On the other hand,
\begin{eqnarray*}
\sum_{j=1}^{p_2}\E\left[\left(\phi_j^{(n)}\right)^2\bigg|\mathcal F_{j-1}^{(n)}\right]
&=&\frac{1}{p_1p_2}\sum_{j=1}^{p_2}\sum_{i=1}^{p_1}\sum_{k=1}^{p_1}a_{ij}a_{kj}\E\left(\psi(\bu_i,\bv_j)\psi(\bu_k,\bv_j)|\bu_i, \bu_k\right)\\
&=&\frac{2}{np_1p_2}\sum_{i=1}^{p_1}\sum_{k=1}^{p_1}b_{ij}\psi(\bu_i,\bu_k)\\
&=&\frac{2}{np_1p_2}\sum_{i=1}^{p_1}b_{ii}\psi(\bu_i,\bu_i)+
\frac{2}{np_1p_2}\sum_{i\neq k}b_{ik}\psi(\bu_i,\bu_k)\\
&:=& A_{1n}+A_{2n},
\end{eqnarray*}
where $b_{ik}=\sum_{j=1}^{p_2}a_{ij}a_{kj}$, $i, k=1,\ldots, p_1$. Considering the variances of $A_{1n}$ and $A_{2n}$,
$\Var(A_{1n})=O(1/n)$ and
\begin{eqnarray*}
\Var(A_{2n})&=&\frac{4}{n^2p_1^2p_2^2}\sum_{i\neq k}\sum_{l\neq s}b_{ik}b_{ls}\E[\psi(\bu_i,\bu_k)
\psi(\bu_l,\bu_s)]\\
&=&\frac{8}{n^2p_1^2p_2^2}\sum_{i\neq k}b^2_{ik}\E[\psi^2(\bu_i,\bu_k)],
\end{eqnarray*}
which is $O(1/n^2)$. Therefore, from the Chebyshev inequality,
$$
\sum_{j=1}^{p_2}\E\left[\left(\phi_j^{(n)}\right)^2\bigg|\mathcal F_{j-1}^{(n)}\right]-\sum_{j=1}^{p_2}\E\left(\phi_j^{(n)}\right)^2
\xrightarrow{p}0,\quad\text{as}\quad (p, n)\rightarrow\infty,
$$
where the second expectation has expression $s_n^2:=2(1-1/n)(1+2/n)\sum_{i=1}^{p_1}\sum_{j=1}^{p_2}a_{ij}^2/(p_1p_2)$.

Next we verify Lyapunov condition by showing that $ B_n=\sum_{j=1}^{p_2}\E(\phi_j^{(n)})^4\rightarrow0$.
From Lemma \ref{lemma1} and the law of iterated expectations,
\begin{eqnarray*}
B_n&=&\frac{1}{p^2_1p^2_2}\sum_{j=1}^{p_2}\sum_{i=1}^{p_1}\sum_{l=1}^{p_1}\sum_{s=1}^{p_1}\sum_{t=1}^{p_1}a_{ij}a_{lj}a_{sj}a_{tj}
\E[\psi(\bu_i,\bv_j)\psi(\bu_l,\bv_j)\psi(\bu_s,\bv_j)\psi(\bu_t,\bv_j)]\\
&=&\frac{1}{p^2_1p^2_2}\sum_{j=1}^{p_2}\sum_{i=1}^{p_1}a^4_{ij}
\E[\psi^4(\bu_i,\bv_j)]+\frac{3}{p^2_1p^2_2}\sum_{j=1}^{p_2}\sum_{i\neq s}a^2_{ij}a^2_{sj}
\E[\psi^2(\bu_i,\bv_j)\psi^2(\bu_s,\bv_j)]\\
&=&\frac{1}{p^2_1p^2_2}\sum_{j=1}^{p_2}\sum_{i=1}^{p_1}a^4_{ij}
\E[\psi^4(\bu_i,\bv_j)]+\frac{12}{n^2p^2_1p^2_2}\sum_{j=1}^{p_2}\sum_{i\neq s}a^2_{ij}a^2_{sj}
\E[\psi^2(\bv_j,\bv_j)],
\end{eqnarray*}
which is $O(1/p_2)$ as $(p, n)\rightarrow\infty.$

Notice that $\hat \gamma_{xx}$ and $\hat \gamma_{yy}$ are unbiased and consistent estimators of $\gamma_{xx}$ and $\gamma_{yy}$, respectively.
The statistic $\hat s_n^2:=2(1-1/n)(1+2/n)\hat \gamma_{xx}\hat \gamma_{yy}/(p_1p_2)$ is also an unbiased and consistent estimator of $s_n^2$ under the null hypothesis, therefore
$$
\frac{n}{\sqrt{2k_n}}\frac{\hat\gamma_{xy}}{\sqrt{\hat \gamma_{xx}\hat \gamma_{yy}}}=\frac{1}{\hat s_n}\sum_{j=1}^{p_2}\phi_{j}^{(n)}\xrightarrow{d} N(0,1),\quad\text{as}\quad (p, n)\rightarrow\infty.
$$

\end{document}